\documentclass[a4paper,11pt]{article}

\usepackage{jheppub}
\usepackage{amsmath,amsfonts,amssymb}
\usepackage{xcolor}
\usepackage{graphicx}
\usepackage{float}
\usepackage{amsthm}
\usepackage{appendix}
\usepackage{blkarray}
\usepackage{multirow}
\usepackage{changes}
\usepackage[ruled,vlined]{algorithm2e}
\usepackage{comment}

\theoremstyle{definition}
\newtheorem{definition}{Definition}[section]
\newtheorem{theorem}{Theorem}[section]
\newtheorem{lemma}{Lemma}[section]
\newtheorem{proposition}{Proposition}[section]
\newtheorem{remark}{Remark}[section]
\newtheorem{corollary}{Corollary}[section]

\numberwithin{equation}{section}

\DeclareMathOperator{\Tr}{Tr}

\title{Non-trivial Area Operators Require Non-local Magic}

\author[1,2,3]{ChunJun Cao}
 \affiliation[1]{Joint Center for Quantum Information and Computer Science, University of Maryland, College Park, MD 20742, USA}
 \affiliation[2]{Institute for Quantum Information and Matter, California Institute of Technology, Pasadena, CA
91125, USA}
\affiliation[3]{Department of Physics, Virginia Tech, Blacksburg, VA, 24061, USA}
\emailAdd{cjcao@vt.edu}

\abstract{We show that no stabilizer codes over any local dimension can support a non-trivial area operator for any bipartition of the physical degrees of freedom even if certain code subalgebras contain non-trivial centers. This conclusion also extends to more general quantum codes whose logical operators satisfy certain factorization properties, including any complementary code that encodes qubits and supports transversal logical gates that form a nice unitary basis. These results support the observation that some desirable conditions for fault tolerance are in tension with emergent gravity and suggest that non-local ``magic'' would play an important role in reproducing features of gravitational back-reaction and the quantum extremal surface formula. We comment on conditions needed to circumvent the no-go result and examine some simple instances of non-stabilizer codes that do have non-trivial area operators.}

\begin{document} 
\maketitle

\section{Introduction}
The intriguing connection between the anti-de Sitter/conformal field theory (AdS/CFT) correspondence and quantum error correcting code (QECC) has led to numerous developments in quantum gravity and quantum information\cite{Almheiri:2014lwa}. A key property in AdS/CFT reproduced by QECC models is the Ryu-Takayanagi (RT) formula \cite{Ryu:2006bv,FLM,Hubeny:2007xt}, which is crucial in understanding how geometry and gravity can emerge from quantum entanglement\cite{Faulkner:2013ica,Faulkner:2017tkh}.  In a seminal work \cite{Harlow:2016vwg}, Harlow shows that all quantum erasure correction codes with encoding map $V:\mathcal{H}_L\rightarrow \mathcal{H}$ satisfying complementary recovery produce an RT-like relation such that for any logical state $\rho\in L(\mathcal{H}_L)$, a finite dimensional von Neumann algebra $M_L$ on $\mathcal{H}_L$ and a factorization $\mathcal{H}=\mathcal{H}_A\otimes \mathcal{H}_{\bar{A}}$,
\begin{equation}
    S(\Tr_{\bar{A}} [V\rho V^{\dagger}]) = S(M_L; \rho) +\Tr[\mathcal{L}_A\rho],
    \label{eqn:RT}
\end{equation}
where $\mathcal{L}_A\in L(\mathcal{H}_L)$ is defined to be the \emph{area operator}. 
Analogous to \cite{FLM},  $S(\Tr_{\bar{A}} [V\rho V^{\dagger}])$ corresponds to the entropy of a CFT subregion $A$, $S(M_L; \rho)$ plays the role of the bulk entropy associated with the entanglement wedge of $A$ and $\Tr[\mathcal{L}_A\rho]$ corresponds to the area of the extremal surface  homologous to $A$. In particular, it is argued \cite{Engelhardt:2014gca,Almheiri:2014lwa,Harlow:2016vwg,Cao:2017hrv,Akers:2018fow,Dong:2018seb,Caonote} that in order to emerge effects analogous to gravity and to describe smooth geometries, $\mathcal{L}_A$ needs to be non-trivial, i.e., not proportional to the logical identity.
Although not all quantum codes have a natural geometric description, where the notion of ``area'' is more form than substance, it is crucial to understand what in quantum codes gives rise to non-trivial $\mathcal{L}_A$, and by extension, provides a necessary condition in emerging spacetime and gravity from quantum information. This is the question we tackle in our work.

Many popular bottom-up constructions of holographic codes\cite{Pastawski:2015qua,Harris:2018jfl,Jahn:2019nmz,steinberg2023holographic} have concentrated on the stabilizer formalism. There are several clear advantages to this approach, (1) the framework is mature and well-understood, (2) the performance of the resulting code and many of its properties are classically efficiently simulable, especially under ``semi-classical'' noise like Pauli error channels, (3) code properties can be engineered and fine-tuned, and (4) they are practically implementable on quantum devices compared to (Haar) random codes or other models based on random tensor networks (RTN)\cite{Hayden:2016cfa}. This makes holographic stabilizer codes prime candidates for building holographic toy models, for teasing out the connections between particular code properties and AdS/CFT, and for simulating quantum gravity toy models on near-term quantum devices.

It is known \cite{Harlow:2016vwg,Donnelly:2016qqt,Cao:2020ksw,Pollack:2021yij} that stabilizer codes can admit non-trivial centers in their code subalgebra, of which $\mathcal{L}_A$ is an element. However, this does not imply the area operator, which is an element of the center, is non-trivial. Common lore often points to the fact that stabilizer codes and RTNs have flat entanglement spectrum, and thus a trivial area operator. However, it should be noted that spectral (non-)flatness alone is not sufficient to show whether the area operators are trivial or not. For example, it is well-known that stabilizer codes with trivial area operators can produce non-flat entanglement spectrum if they encode non-stabilizer logical states\footnote{In fact, the encoded logical state doesn't even have to be entangled. See for example the 4-qubit Bacon-Shor code in Section 3.1 of \cite{Cao:2020ksw}.}. On the other hand, \cite{Cao:2020ksw,Pollack:2021yij} constructed examples of codes with non-trivial area operators but  the codewords retain a flat entanglement spectrum\footnote{See, for example, Sections 4.1 and 7.2 of \cite{Cao:2020ksw} and Section 5.2 of \cite{Pollack:2021yij}.}. Therefore, it is conceivable that although stabilizer codes do not reproduce certain features of holography such as correlation functions, these codes may yet have the capacity to accommodate non-trivial area operators, and thus some analogous form of gravity.  However, all of the examples that successfully constructed non-trivial area operators are, in fact, non-stabilizer codes \cite{Cao:2020ksw,Pollack:2021yij}. This led the authors of \cite{Pollack:2021yij} to question whether there are fundamental obstructions in stabilizer codes to having non-trivial $\mathcal{L}_A$. Here we formalize their conjecture and prove that no stabilizer code can have non-trivial area operators even if their code subalgebras have non-trivial centers. This no-go result also generalizes to local unitary deformations of stabilizer codes as well as classes of non-stabilizer codes with transversal logical gates that form a nice unitary basis.

These findings highlight the importance of long range non-stabilizerness\cite{NLmagic,White:2020zoz,Bao:2022mkc,Liu:2020yso} in understanding the emergence of gravitational features. It is well-known that when combined with Clifford operations, non-stabilizerness, or magic, is necessary and sufficient for universal quantum computation. In light of our no-go result,  this implies that any code that supports a non-trivial area operator must also contain some amount of ``non-local'' magic that cannot be removed simply from local unitaries. 
From the point of view of quantum gravity and its simulations, they also motivate code design beyond the Pauli stabilizer formalism and underscore the tension between gravity and fault tolerance.

We refer the readers to \cite{Harlow:2016vwg} and \cite{Pollack:2021yij} for details on the RT formula in QECCs. In this article, we only work with codes and von Neumann algebras on finite dimensional Hilbert spaces.  Note that the term ``stabilizer code'' commonly refers to a Pauli stabilizer code in literature where the code subspace is defined by the simultaneous $+1$ eigenspace of all elements in an abelian subgroup of the Pauli group over $n$ qudits with local dimension $q$. Here we adopt the same convention; as such, non-stabilizer codes in this work also include codes like \cite{XS,XP} which extend the Pauli stabilizer formalism\footnote{Strictly speaking, they are also stabilizer codes except the stabilizer group is no longer an abelian Pauli subgroup. However, their codewords in the computational basis need not be stabilizer states.}.

\section{No-go Theorems}
\begin{definition}\label{def:nicebasis}
    Consider the set of unitary operators over a finite dimensional Hilbert space $\mathcal{H}$, we say $\mathcal{E}$ is a unitary basis if it contains the identity operator and $\Tr[EF^{\dagger}]=0$ for $E,F\in \mathcal{E}, E\ne F$. Furthermore, $\mathcal{E}$ is a \emph{nice error basis} \cite{nice_err_basis1,nice_err_basis2} if
    $$EF=\omega(E,F)FE$$
    for all $E,F\in \mathcal{E}$, where $\omega(E,F)\in \mathbb{C}$ satisfies $\omega(E,F)^r=1$ for some $r$.
\end{definition}
For example, the set of Pauli operators (including identity) forms a nice error basis.

\begin{theorem}\label{thm:1}
Let $\mathcal{C}\subseteq \mathcal{H}$ be a stabilizer code with $M_{\mathcal{C}}$ the associated von Neumann algebra on $\mathcal{C}$. Any area operator ${\mathcal{L}}_{A}\in M_{\mathcal{C}}$ defined in this code must be proportional to the identity.
\end{theorem}
\begin{proof}
Consider any bipartition  of the physical degrees of freedom such that $\mathcal{H}=\mathcal{H}_A\otimes \mathcal{H}_{\bar{A}}$ and let $M_A$ be the code subalgebra supported on $A$. Let $\{|\widetilde{\alpha,ij}\rangle\}$ be a set of orthonormal basis  induced by $M_A$ in the Wedderburn decomposition\footnote{Note that the most general Wedderburn decomposition has $\mathcal{C}=(\bigoplus_{\alpha}\mathcal{H}_{a^{\alpha}}\otimes \mathcal{H}_{\bar{a}^{\alpha}})\oplus\mathcal{H}_0$. Here we drop $\mathcal{H}_0$ for simplicity as it does not affect the argument.}, which spans $\mathcal{C}=\bigoplus_{\alpha}\mathcal{H}_{a^{\alpha}}\otimes \mathcal{H}_{\bar{a}^{\alpha}}$ such that $\forall\tilde{O}_A\in M_A$, $\tilde{O}_A=\bigoplus_{\alpha}O_{a^{\alpha}}\otimes I_{\bar{a}^{\alpha}}$ , and for operators in the commutant $\forall \tilde{O}'_A\in M'_{{A}}$, $\tilde{O}'_A=\bigoplus_{\alpha}I_{{a}^{\alpha}}\otimes O_{\bar{a}^{\alpha}}$. 

Recall from \cite{Pollack:2021yij} Lemma C.1 that all stabilizer codes satisfy complementary recovery, that is, one can identify a decomposition $\mathcal{H}_A=\bigoplus_{\alpha}\mathcal{H}_{A_1^{\alpha}}\otimes\mathcal{H}_{A_2^{\alpha}}$,$\mathcal{H}_{\bar{A}}=\bigoplus_{\alpha}\mathcal{H}_{\bar{A}_1^{\alpha}}\otimes\mathcal{H}_{\bar{A}_2^{\alpha}}$
and unitary operators $U_A,U_{\bar{A}}$ such that
\begin{equation}
    |\widetilde{\alpha,ij}\rangle = U_AU_{\bar{A}} \Big(|\alpha,i\rangle_{A_1^{\alpha}}|\alpha,j\rangle_{\bar{A}_1^{\alpha}}|\chi_{\alpha}\rangle_{A_2^{\alpha}\bar{A}_2^{\alpha}}\Big)
\end{equation} 
for some state $|\chi_{\alpha}\rangle$ that captures the residual entanglement between $A$ and $\bar{A}$ which we would then associate with the area of the RT surface.
By the definition of complementary recovery\cite{Harlow:2016vwg}, we know that any element of $M_A'$ is supported on $\bar{A}$, i.e., $M_A'\subseteq M_{\bar{A}}$. Since any operator supported on $\bar{A}$ must commute with those on $A$, we have $M_{\bar{A}}\subseteq M_A'$. It follows that $M_A'=M_{\bar{A}}$.

Let $\chi_{\alpha}=\Tr_{\bar{A}_2^{\alpha}}[|\chi_{\alpha}\rangle\langle\chi_{\alpha}|]$, and the area operator is defined as $\mathcal{L}_A=\bigoplus_{\alpha}S(\chi_{\alpha})\tilde{I}_{a^{\alpha}\bar{a}^{\alpha}}\in Z(M_A)$.
If the center $Z(M_A)$ is trivial, i.e. contains only elements proportional to the logical identity, then so is the area operator. If the center is non-trivial, (e.g.\cite{Donnelly:2016qqt,Cao:2020ksw,Pollack:2021yij}) then 
it can be expressed in the block diagonal form with distinct values of $\alpha$. Each such block is often referred to as an $\alpha$-block in relevant literature in holography.  
We now show that the eigenvalues $S(\chi_{\alpha})$ must be identical for all $\alpha$'s in stabilizer codes.

Let $\tilde{P}$ be a logical Pauli operator such that $\tilde{P}\in M_{\mathcal{C}}$ and $[\tilde{P},\mathcal{L}_{A}]\ne 0$. Note that the logical Pauli operators form a complete operator basis for $M_{\mathcal{C}}$ in a stabilizer $[[n,k,d]]_q$ code over any local dimension $q$. If no such $\tilde{P}$ exists, then $\mathcal{L}_A\in Z(M_{\mathcal{C}})$, leading to a trivial area operator.

When such a $\tilde{P}$ can be found, we recall that the Pauli basis is a nice unitary basis. Hence by Lemma~\ref{lemma:2}, there must exist some logical Pauli operator $\tilde{Q}_i$ such that it and $\mathcal{L}_A$ can be simultaneously diagonalized. 
Let $\{|\tilde{\psi}_{\alpha}\rangle\}$ be such an eigenbasis over $\mathcal{C}$ for the Pauli operator $\tilde{Q}_i$ where each basis state lives in a particular $\alpha$ block.
For codes that are complementary recoverable, there exist recovery unitaries $U_A, U_{\bar{A}}$ such that, 
\begin{align}
    U_A U_{\bar{A}}\tilde{P}|\tilde{\psi}_{\alpha}\rangle= P' |\psi_{\alpha}\rangle_{A_1^{\alpha}\bar{A}_1^{\alpha}} \otimes |\chi_{\alpha}\rangle_{A_2^{\alpha} \bar{A}_2^{\alpha}},
\end{align}
where $P'$ is a Pauli operator. Here we have used the fact that $U_A,U_{\bar{A}}$ are Clifford unitaries, and hence maps Pauli operators to Pauli operators. We will drop the $A_i^{\alpha}\bar{A}_i^{\alpha}$ labels to reduce clutter.

 Because each state $|\tilde{\psi}_{\alpha}\rangle$ is an eigenstate of a Pauli operator, by Lemma \ref{lemma:1}, we know that $\tilde{P}|\tilde{\psi}_{\alpha}\rangle=|\tilde{\psi}'_{\alpha'}\rangle$ where $|\tilde{\psi}'_{\alpha'}\rangle$ is an(other) eigenstate of $\mathcal{L}_A$. Therefore, we must have 
 
 \begin{equation}
     P' (|\psi_{\alpha}\rangle\otimes |\chi_{\alpha}\rangle) = P'_1 |\psi_{\alpha}\rangle\otimes P'_2|\chi_{\alpha}\rangle=|\psi'_{\alpha'}\rangle\otimes |\chi'_{\alpha'}\rangle,
     \label{eqn:Pprime}
 \end{equation}
 such that we are taken to a potentially different state on an alpha block with potentially a different eigenvalue. However, because both $P'_1, P'_2$ are tensor products of Pauli operators, they do not alter the entanglement of $\chi_{\alpha}$ across $A_2^{\alpha},\bar{A}_2^{\alpha}$, i.e., $S(\chi'_{\alpha'})=S(\chi_{\alpha})$ where we have defined $\chi_{\alpha}=\Tr_{\bar{A}_2}[|\chi_{\alpha}\rangle\langle\chi_{\alpha}|]$ and $\chi'_{\alpha'}=\Tr_{\bar{A}_2}[|\chi'_{\alpha'}\rangle\langle\chi'_{\alpha'}|]$.
As this holds for all of the above basis states, it follows that for any state $\tilde{\rho}\in L(\mathcal{C})$, $\Tr[\tilde{\rho}\mathcal{L}_A]=\sum_{\alpha}p_{\alpha}S(\chi_{\alpha})$ is conserved under action generated by such $\tilde{P}$'s. In other words,
$[\mathcal{L}_A,\tilde{P}] =0,$ which is a contradiction. Therefore the area operator must be trivial.
\end{proof}

The proof relies on two crucial elements --- that $Z(M_{\mathcal{C}})$ is trivial, as implied by the definition of the code,  and that Pauli operators are mapped to Pauli operators by Clifford unitaries. 
Note that $Z(M_{\mathcal{C}})$ can be trivial while $Z(M_A)$, the center of the code subalgebra recoverable on $A$, need not be. This is the case, for example, for the $[[4,1,2]]$ stabilizer code \cite{Cao:2020ksw}. 
Intuitively, 
because recovery unitaries $U_{A},U_{\bar{A}}$ are Cliffords in a stabilizer code, which map logical Pauli operators to Pauli operators, their resulting actions on $|\chi_{\alpha}\rangle$ do not change the entanglement across  $A_2^{\alpha}$ and $\bar{A}_2^{\alpha}$, which is what gives rise to the RT contribution in the entropy formula (\ref{eqn:RT}). Such conditions are generally broken by non-(Pauli)-stabilizer codes, where $M_{\mathcal{C}}$ can take on more general forms and the decoding unitaries need not map tensor product of operators to tensor product of operators. 

At the same time, we note that such a no-go result continues to hold for a more general class of codes as long as the requisite logical Pauli operator $\tilde{P}$ is mapped to an operator that preserve the entanglement across $A_2^{\alpha}$ and $\bar{A}_2^{\alpha}$, for which the tensor product of Pauli's is sufficient but not necessary. For example,

\begin{proposition}\label{prop:LU}
    Suppose a quantum code is unitary equivalent to a stabilizer code $\Pi_c$, \emph{i.e.}, $\Pi=Q\Pi_c Q^{\dagger}$ and $Q=Q_{A}\otimes Q_{\bar{A}}$, then it has a trivial area operator for the bipartition $A,\bar{A}$.
\end{proposition}

\begin{proof}
With deformation $Q$, it is clear that the codewords and logical Pauli's are shifted such that $|\tilde{\psi}\rangle\rightarrow |\tilde{\psi}^Q\rangle=Q|\tilde{\psi}\rangle$ and $\tilde{P}\rightarrow \tilde{P}_Q=Q\tilde{P}Q^{\dagger}$.

Repeating the proof of Theorem~\ref{thm:1}, we simply need to show that no logical Pauli alters the entanglement of $|\chi_{\alpha}\rangle$.  From the definition of $\Pi$, we know that the code is still complementary with respect to the $A,\bar{A}$ bipartition such that $U'_{A}U'_{\bar{A}}|\tilde{\psi}_{\alpha}^Q\rangle =|\psi_{\alpha}\rangle\otimes |\chi_{\alpha}\rangle$ where one construction is simply $U'_{A}=U_AQ_A^{\dagger}, U'_{\bar{A}}=U_{\bar{A}}Q_{\bar{A}}^{\dagger}$ and $U_A,U_{\bar{A}}$ are Clifford unitaries. Note that the complementary recovery map generally does not take this nice form if $Q$ does not factorize. Therefore
\begin{equation}
    U'_AU'_{\bar{A}}\tilde{P}_Q|\tilde{\psi}^Q_{\alpha}\rangle = U_A U_{\bar{A}} Q^{\dagger} Q\tilde{P}Q^{\dagger} Q|\tilde{\psi}_{\alpha}\rangle = P'|\psi_{\alpha}\rangle|\chi_{\alpha}\rangle
\end{equation}
where $P'$ is again a Pauli operator that does not alter the entanglement of $|\chi_{\alpha}\rangle$ between $A_2$ and $\bar{A}_2$. 
\end{proof}

Let us now comment on the connection between area operator and non-stabilizerness. The non-local magic of a state is introduced by \cite{NLmagicHamma} and the definition is given below for convenience.  We refer interested readers to~\cite{NLmagicHamma} for more detailed definitions and discussions.
\begin{definition}
    The $m$-partite non-local magic of a state is given by 
    $$\mathcal{M}^{(m)}_{NL}(|\psi\rangle_{A_1\dots A_m}) = \min_{U=U_{A_1}\otimes \dots\otimes U_{A_m}}\mathcal{M}(U|\psi\rangle),$$
    where $\mathcal{M}(|\psi\rangle)\geq 0$ is a faithful measure of magic\footnote{The most general definition by \cite{NLmagicHamma} does not require the magic measure to be faithful, but all bounds and implementations used faithful measures.} such that $\mathcal{M}(|\psi\rangle)=0$ if and only if $|\psi\rangle$ is a stabilizer state.
\end{definition}

Intuitively, a magic measure $\mathcal{M}$ computes the ``distance'' between the state $|\psi\rangle$ and the nearest stabilizer state. The non-local magic takes it one step further by first removing any non-stabilizerness contained in the subsystems $A_1,\dots, A_m$ then computes the remaining magic. In other words, it quantifies the ``distance'' between the state $|\psi\rangle$ and the nearest stabilizer state up to any $m$-partite local unitary transformation $U$. In particular, the non-local magic vanishes if and only if the state is equal to a stabilizer state up to local unitary deformations. Therefore, in the same way,
\begin{definition}
    A code with trivial $m$-partite non-local magic over a partition $A_1,\dots, A_m$ is one that can be unitarily deformed to a stabilizer code for some $U=U_{A_1}\otimes \dots\otimes U_{A_m}$.
\end{definition}

\begin{corollary}
    Given an $m$-partition $\mathcal{K}=\{A_1,\dots, A_m\}$ of the physical qudits in an $[[n,k,d]]_q$ code, if the code has trivial $m$-partite non-local magic, then the area operator is trivial for any bipartition $A, \bar{A}$ where $A$ is the union of elements in any subset $\mathcal{K}_A\subset \mathcal{K}$. 
\end{corollary}
\begin{proof}
   Trivial non-local magic implies the existence of local unitary deformation $\bigotimes_{i\in\mathcal{K}}U_i$ that does not entangle the partitions. We simply rewrite $Q_A=\bigotimes_{i\in \mathcal{K}_A} U_i$ and $Q_{\bar{A}}=\bigotimes_{i\in \mathcal{K}\setminus\mathcal{K}_A} U_i$. Then the conclusion follows from Proposition~\ref{prop:LU}.
\end{proof}

\begin{remark}
   These codes obtained by local unitary deformations are trivial extensions of the stabilizer code. They have local magic, but not non-local magic, \emph{i.e.}, non-stabilizerness that cannot be removed by acting ``local'' unitary gates. 
    Non-local magic is necessary for (holographic) codes with non-trivial area operators. It also plays a crucial role in generating non-flatness in the entanglement spectrum expected of a CFT\cite{NLmagicHamma}.
\end{remark}

However, to construct non-trivial area operators, it is insufficient to simply consider just any non-stabilizer code. Beyond these trivial extensions, there can also be more general codes that have trivial area operators but need not be stabilizer codes as long as certain product structures are preserved under the recovery process. 

\begin{theorem}\label{thm:2}
    For any quantum code $\mathcal{C}$ with an associated code subalgebra $M_{\mathcal{C}}$, let  $\mathcal{E}_{\mathcal{C}}$ be a nice error basis for the logical operators. Suppose $\mathcal{C}$ satisfies complementary recovery with respect to bipartition $\mathcal{H}=\mathcal{H}_A\otimes\mathcal{H}_{\bar{A}}$ and $Z(M_{\mathcal{C}})$ is trivial then it admits only a trivial area operator if all unitary logical operators $\tilde{E}\in\mathcal{E}_\mathcal{C}$ are mapped to 
$$U\tilde{E}U^{\dagger}=O_{A_1\bar{A}_1}\otimes O_{A_2}\otimes O_{\bar{A}_2},$$
where $U=U_AU_{\bar{A}}$ are unitary recovery maps, 
\end{theorem}

\begin{proof}
The proof is the same as that for the stabilizer code in Theorem~\ref{thm:1}, except we instead have $\tilde{P},\tilde{Q}_i$ in (\ref{eqn:Pprime}) be elements of the nice error basis. In the same way, $P'$s are no longer Pauli operators in general. By assumption, it maps $\tilde{P}$ to $O_{A_1\bar{A}_1}\otimes O_{A_2}\otimes O_{\bar{A}_2}$. This unitary operator does not entangle $A_2$ with $\bar{A}_2$. Hence the rest of the proof follows.  
\end{proof}

\begin{corollary}\label{coro:2}
If $\mathcal{C}$ is complementary, $Z(M_{\mathcal{C}})$ is trivial and the logical unitary operators in $\mathcal{E}_{\mathcal{C}}$ are transversal, then the code has trivial area operators for all bipartitions.
\end{corollary}
\begin{proof}
    If $\tilde{E}\in\mathcal{E}_{\mathcal{C}}$ is transversal, i.e., $\tilde{E}=\bigotimes_i O_i$, then for any bipartition $\mathcal{H}=\mathcal{H}_A\otimes \mathcal{H}_{\bar{A}}$, $U_{A}U_{\bar{A}} \tilde{E} (U_{A}U_{\bar{A}})^{\dagger} = O_{A}\otimes O_{\bar{A}}$. Since $\tilde{E}$ is also a logical operator, it satisfies the general form required by Theorem~\ref{thm:2} that does not entangle $A_2$ and $\bar{A}_2$, hence the area operator is trivial. %This includes Corollary~\ref{coro:1} as a special case.
\end{proof}
Codes covered by Corollary~\ref{coro:2} include some non-(Pauli-)stabilizer codes such as regular XP codes and XS codes \cite{XP,XS}; hence it is clear that there are non-stabilizer codes (that are not LU equivalent to stabilizer codes) which also have trivial area operators.  As transversal operators are generally desirable for fault tolerance, we see that the emergence of gravity is in tension with it. 
\begin{remark}
    In fact, for any code that encodes the tensor product of qubits, its code subalgebra is generated by the nice error basis such as the Pauli group elements, and hence would have trivial $Z(M_{\mathcal{C}})$. By the above Corollary, although these codes with transversal Pauli's (or other unitary gates) are generally desirable for fault-tolerant quantum computation, they are not so for emergent gravity as they produce trivial area operators.
\end{remark}

Another instance of non-stabilizer code with trivial area operator may be the (Haar) random tensor networks (RTN) in the large bond dimension limit\cite{Hayden:2016cfa} as they describe fixed area states\cite{Akers:2018fow,Dong:2018seb}. However, RTN models are approximate QECCs and said result holds in the infinite bond dimension limit. Therefore, strictly speaking they are beyond the scope of codes we discuss in this work. On the other hand, the subclass of random stabilizer tensor networks at finite bond dimensions that are powers of $q$ can be rephrased as stabilizer codes discussed in this work, and hence the no-go theorem still applies.

Interestingly, to produce non-trivial area operators, it appears that we need to consider criteria that are the opposite of what we usually use to construct fault-tolerant gates. Nevertheless, some level of fault-tolerance seems to be retained  in AdS/CFT \cite{threshold} despite this difficulty to incorporate transversal gates. See also \cite{Woods,Faist,Cree:2021rxi}. 

\begin{remark}\label{rmk:erasure}
Any code with a trivial $Z(M_{\mathcal{C}})$ but a non-trivial area operator for all bipartitions 
 $\{A,\bar{A}: |A|,|\bar{A}|\ne 0\}$ cannot have distance $d>1$.
\end{remark}
\begin{proof}
Since area operator is non-trivial, but $Z(M_{\mathcal{C}})$ is, there must exist some logical operator which takes states in one alpha block to another. Let this operator be $\tilde{O}$ and $[\mathcal{L}_A,\tilde{O}]\ne 0$. Recall that $\mathcal{L}_A\in Z(M_A)$, which has support on either $A$ or $\tilde{A}$ entirely, this means that $\tilde{O}$ must have support on both $A$ and $\bar{A}$. If this is true for all bipartitions, then there will always be a subalgebra generated by $\tilde{O}$ that is not recoverable after the erasure of $\bar{A}$. Hence the erasure correction cannot be exact. 
\end{proof}

Although the assumptions of Remark~\ref{rmk:erasure} are somewhat restrictive, the outcome is not terribly surprising in AdS/CFT because  any erasure correction is only expected to be approximate at finite $N$.
A similar, but less restrictive, distance upper bound is expected in order to reproduce non-trivial boundary 2-point functions\cite{HMERA}. Here the adversarial distance is not a particularly informative measure of the code property because large erasures may still be correctable up to small errors though the recovery is not exact.

It is a common lore that a non-trivial area operator is tied to soft-hair or edge modes such that the ``physical'' Hilbert space in a system with gauge symmetry is not factorizable unlike the logical subspace of a stabilizer code\cite{Buividovich:2008gq,Donnelly:2011hn,Casini:2013rba,Lin:2017uzr}. Such a system may be identified in a stabilizer code by considering a subspace of $\mathcal{C}$ that satisfies additional gauge constraints such that ``emergent gauge symmetry'' is present in a (low energy) subspace $C'$\cite{Levin:2004mi,Dolev:2021ofc,Qi:2022lbd}. Let us examine whether our no-go theorem puts any constraint on such codes.

For example, what if we do not consider proper stabilizer codes, but an induced version with a more  restricted form of the code subalgebra $\mathcal{N}_A\subset \mathcal{M}_A$? Such an operator algebra QECC generally need not be a stabilizer code. Ref. \cite{Pollack:2021yij} points out that complementary recovery is so restrictive that any von Neumann algebra $M_A$ satisfying such properties must be unique. This may be somewhat undesirable because although one may still identify one-sided area operators under $\mathcal{N}_A$ which may be non-trivial, the resulting code cannot satisfy complementary recovery. 

It is also possible to constrain the entire code subalgebra $\mathcal{N}\subset\mathcal{M}$, which amounts to (generalized) concatenation where one effectively defines another code over the $k$ logical qubits/qudits. The no-go theorem does not constrain codes obtained this way, except when the concatenated codes are stabilizer codes. It is known that (generalized) concatenation of stabilizer codes remain stabilizer codes\cite{QL}, hence such combination would yield trivial area operators. For example, by gauging the bulk and constructing a $\mathbb{Z}_2$ gauge theory over $k$ bulk qubits in any holographic stabilizer toy models would still yield trivial area operators. Therefore, while gauging would be useful in generating non-trivial centers, it alone is insufficient to guarantee the existence of non-trivial ``area operators''. This constitutes another reason for considering non-stabilizer codes.

\section{Codes with non-trivial area operators}
Although we are still far from understanding how codes with desirable non-trivial area operators can be systematically created from the bottom up, we can build up some intuition through examples. 
One such prescription for introducing a non-trivial area operator is shown in (Fig.~\ref{fig:circuit}). Each wire $A_1^{\alpha},\bar{A}_1^{\alpha},A_2^{\alpha}$ or $\bar{A}_2^{\alpha}$ in Fig.~\ref{fig:circuit}a is a tensor factor of a Hilbert space which need not have the same dimensions. However, distinct from the usual circuit diagrams, it is understood that these Hilbert spaces are also direct summed over $\alpha$. To distinguish from the usual single and double wires in circuit diagrams that represent tensor factors and classical operators, we use a fuzzy line to denote the implicit direct sum over different $\alpha$ sectors. Solid line on $U_A,U_{\bar{A}}$ denotes action on factorized Hilbert spaces as usual. Circle denotes controlled system. Controlled-controlled gate here is used to indicate the unitary action $$U = \sum_{\alpha}\sum_{i,j} |\alpha;ij\rangle\langle \alpha;ij|\otimes U_{\alpha},$$ which can depend on the $\alpha$-block on which the state of the wire lives. $U_{\alpha}$ acts on the subspace $\mathcal{H}_{A_2^{\alpha}\bar{A}_2^{\alpha}}$. Operationally, it is more convenient to construct such codes with a circuit over tensor factors like physical qubits; we may do so by introducing additional tensor factors to keep track of the $\alpha$ blocks and embedding the circuit into one over a factorizable Hilbert space (Fig.~\ref{fig:circuit}b). Each wire represents a tensor factor, but they need not be of the same dimensionality.

\begin{figure}
    \centering
    \includegraphics[width=0.55\linewidth]{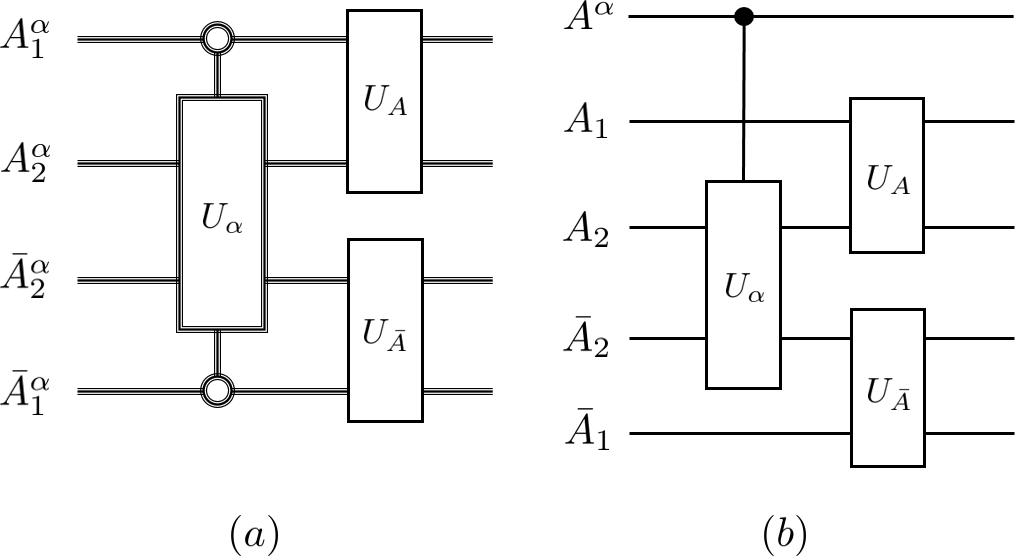}
    \caption{(a) A unitary circuit for preparing a code with non-trivial area operator in an $A$-$\bar{A}$ bipartition. (b) A circuit that achieves the left via embedding into a factorizable Hilbert space. The controlled gate is given by $\sum_{\alpha}|\alpha\rangle\langle\alpha|\otimes U_{\alpha}$. }
    \label{fig:circuit}
\end{figure}

For a more explicit example, consider the $[[4,1,1]]$ skewed code\cite{Cao:2020ksw} with trivial $A_1,\bar{A}_1$ defined by the encoding isometry $V=|0000\rangle\langle 0|+(|1001\rangle+|0110\rangle)\langle 1|/\sqrt{2}$. The overall encoding tensor is represented as Fig.~\ref{fig:4qubit}b. 
The codewords $|\tilde{0}\rangle, |\tilde{1}\rangle$ have entanglement structures that are invariant under rotations by construction. Using the protocol above, we can now construct this code with an explicit encoding circuit (Fig.~\ref{fig:4qubit}a). A simple way to verify the distance of the code is through the quantum weight enumerators\cite{SL}. The enumerators of this code are 
\begin{align}
    A(z)=&1+z+2z^2+z^3+3z^4,\\
    B(z)=&1+2z+10z^2+10z^3+9z^4.
\end{align}
Note that $B(z)-A(z)$ has non-zero coefficient at $z^1$, hence $d=1$.
It is also simple enough that we can verify the distance by observation: the Knill-Laflamme condition is violated for an error of weight 1, i.e., $\langle +|Z_1|-\rangle =1/2\not\propto\delta_{+-}=0.$

\begin{figure}
    \centering
    \includegraphics[width=0.55\linewidth]{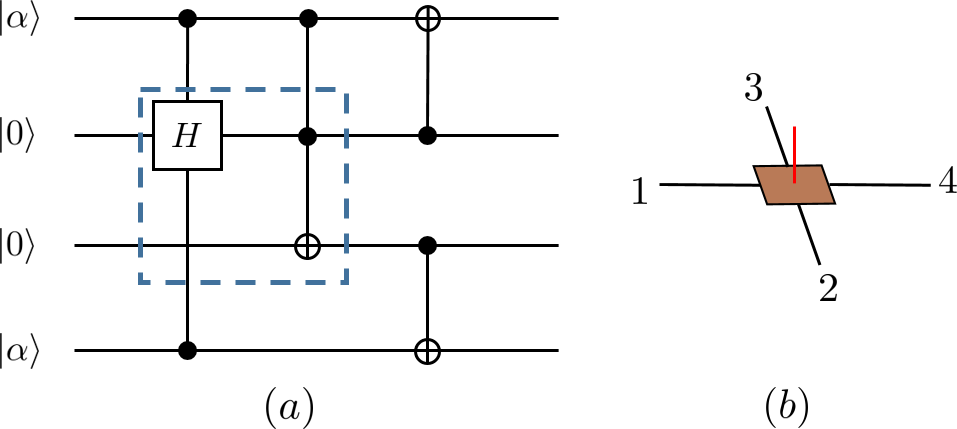}
    \caption{(a) Encoding circuit of the $[[4,1,1]]$ code where the qubits labelled sequentially from top to bottom. The dashed box represent $U_{\alpha}$ while the states $|\alpha\rangle$ correspond to the degrees of freedom in $A^{\alpha}$. (b) The same code represented as an encoding map/tensor with the red line marking the logical input and the black lines marking physical qubits.}
    \label{fig:4qubit}
\end{figure}

Note that the codewords are still stabilized by a Pauli subgroup $\langle ZIZI,IZIZ\rangle$, but other symmetries, which are not Pauli operators, exist. A representation of the logical X operator is
\begin{align}
    \tilde{X}=\frac{1}{\sqrt{8}}(IXIX-ZXIX+IXXI+ZXXI
    +XIIX+XZIX-XIXI+XZXI)
\end{align}
and one can easily read off some representations of $\tilde{Z}$ from the code words as weight-2 $ZZ$ operators, e.g., $Z_1Z_2$. 

As a result, this code is complementary in the 2 qubit by 2 qubit bipartition as long as each partition contains spatially adjacent qubits. One can see this by realizing that each subregion reconstructs the $\tilde{Z}$ subalgebra. It also has a non-trivial area operator ($S(\chi_{0})=0, S(\chi_{1})=1$)\footnote{Also see example G of \cite{Pollack:2021yij} for an atomic code with non-trivial $A_1,\bar{A}_1$. }. It is, however, not complementary recoverable in the 3-1 partition as the single qubit does not support a non-trivial logical operation while the remaining 3 qubits do not reconstruct the ${X}$-code subalgebra. In other words, $M_A:=V^{\dagger} L(\mathcal{H}_A)\otimes I_{\bar{A}} V$ is not a von Neumann algebra. However, the area operator can still be defined on one side.

Another key feature lies in the non-stabilizerness of the code. This is already apparent in the use of controlled-controlled gates, but us further examine it from the perspective of the Choi state of the encoding map

\begin{equation}|V\rangle = V(|00\rangle+|11\rangle)/\sqrt{2}= \frac{1}{\sqrt{2}}|00000\rangle + \frac 1 2 |10011\rangle + \frac 1 2 |01101\rangle.\end{equation}
We compute the robustness of magic (ROM) $\mathcal{R}(\rho)$ \cite{ROM}, which is a faithful measure of magic\footnote{For the sake of consistency, we will use an affine definition of ROM with a $-1$ offset such that $\mathcal{R}(\rho)=0$ for a stabilizer state, as opposed to 1 in some definitions.}. 
 All 1 or 2-site reduced density matrices of the Choi state are already diagonal in the computational basis, and hence are stabilizer states ($\mathcal{R}(\rho_{1,2})=0$). Some subsystems of 3 qubits have non-stabilizerness, though the majority of such subsystems are stabilizer states. The global state, however, has $\mathcal{R}(|V\rangle\langle V|)=0.7476$, hence the state is magical. Heuristically, therefore, magic is indeed distributed non-locally in the system.
This is in contrast with the Choi state of any stabilizer code encoding isometry, which has zero magic for any (sub)system.

To scale up beyond small atomic examples, one can also create more complicated tensor network encoding maps that inherit some non-trivial area operators by gluing these atomic ``legos'' with Bell fusion --- each connected edge in the tensor network is prepared by projecting two physical qubits of the respective seed codes into a Bell state. Some of the code properties can be deduced from a more refined version of operator pushing. However, the encoding maps generated by such tensor networks are generally non-isometric. Hence they are distinct from the conventional quantum codes and cannot always be unitarily encoded.

For instance, by gluing two $[[4,1,1]]$ codes, Fig.~\ref{fig:noniso}a represents an encoding map of a $[[6,2]]$ code. Without loss of generality, one can understand this map as the composition of the encoding isometry of the left code with another map $V'({\psi}):\mathbb{C}^{2}\rightarrow \mathbb{C}^8$ that is dual to the 4-qubit state $|\tilde{\psi}\rangle$ of the right code via the Choi-Jamiolkowski isomorphism. Therefore, when $|\tilde{\psi}\rangle =|\tilde{1}\rangle$ on the right logical qubit, we produce a concatenated code that inherits the non-trivial area operator of the left $[[4,1,1]]$ code. The only difference is coming from the code concatenation where a single physical qubit on the left code is isometrically mapped to 3 physical qubits on the right code. 

For any $|\tilde{\psi}\rangle\ne |\tilde{1}\rangle$, no single-qubit subsystem of the $[[4,1,1]]$ code is maximally mixed, so $V'$ is not an isometry. As a result, logical $X$ operators cannot be reconstructed exactly on the physical qubits as unitary operators. However, some of them can still be unitarily reconstructed in a state-dependent fashion\footnote{Although the encoding map in Fig.~\ref{fig:noniso}b doesn't have a non-trivial kernel like the non-isometric codes defined in \cite{AP2021,Akers:2022qdl}, it still stipulates state-dependent reconstructions for certain operators and causes bulk qubits to ``overlap''\cite{OQ,OQdS,Caonote} in that the commutator of spacelike separated observables is non-vanishing.}.  Explicitly, let $\tilde{X}_1$ be the unitary operator constructed via operator pushing when the logical qubit on the right is fixed to be $|\tilde{1}\rangle$. %Note that $\tilde{X}_1$ is a state-specific construction \cite{AP2021} that has support over all 6 qubits. 
\begin{figure}
    \centering
    \includegraphics[width=0.5\linewidth]{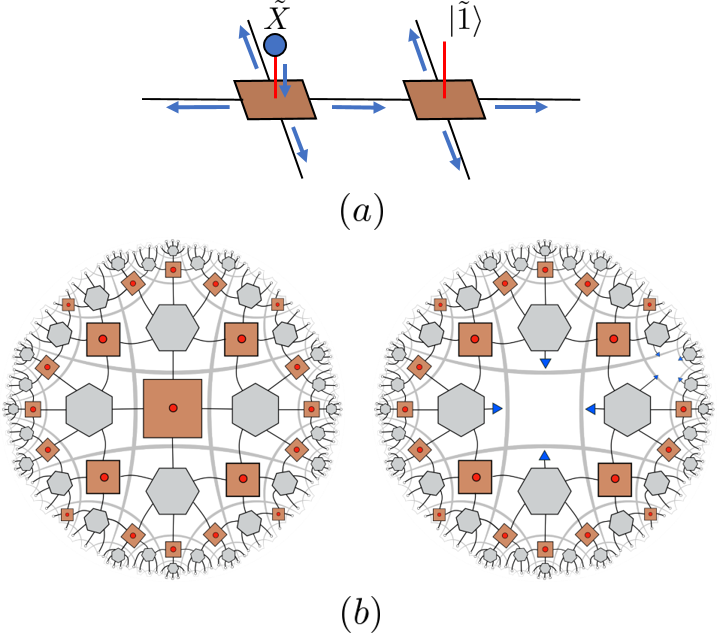}
    \caption{(a) Reconstruction of an operator $\tilde{X}_1$ via operator pushing when the second logical qubit is fixed to $|\tilde{1}\rangle$. (b) The holographic toy model using the $[[4,1,1]]$ atomic lego forms a skewed holographic Bacon-Shor code. Fixing logical legs to $|\tilde{0}\rangle$ alters the entanglement structure and the network connectivity by contracting the tensors with $|0\rangle$s (blue triangles). The grey hexagons are perfect tensors\cite{Pastawski:2015qua}.}
    \label{fig:noniso}
\end{figure}
A state-dependent representation of $\tilde{X}_1\otimes\tilde{I}_2$ for the $[[6,2,1]]$ code (Fig.~\ref{fig:noniso}a) can be written as
\begin{align*}
  \tilde{X}_1\otimes\tilde{I}_2=& \frac{1}{\sqrt{8}}(IXIXXX-ZXIXXX+IXXIII\\
  &+ZXXIII+XIIXXX+XZIXXX\\
  &-XIXIII+XZXIII).  
\end{align*}
It is easy to verify that it is a unitary operator such that $\tilde{X}_1\otimes\tilde{I}_2|\tilde{01}\rangle=|\tilde{11}\rangle$ but $\tilde{X}_1\otimes\tilde{I}_2|\tilde{00}\rangle\ne|\tilde{10}\rangle$. The overlap with the target state is $\langle \tilde{00}|\tilde{X}_1\otimes\tilde{I}_2|\tilde{10}\rangle=1/\sqrt{2}$.

The rotational symmetries in the $[[4,1,1]]$ code entanglement structure makes it a suitable component for holographic tensor networks. By gluing such atomic codes to a hyperbolic tiling\cite{QL}, the RT surface and area can now can take on a more geometric interpretation given by minimal cuts in the graph. For example, one can construct a more sophisticated model which is a non-stabilizer code but has non-trivial area operator. Fig.~\ref{fig:noniso}b is one such construction, which is an instance of the (gauge-fixed) skewed holographic Bacon-Shor code. For certain bipartitions of the boundary degrees of freedom, the code inherits non-trivial area operators from the $[[4,1,1]]$ legos and the entanglement structure as well as the shape of the minimal surface can depend on the bulk state. For example, if all bulk states in Fig.~\ref{fig:noniso}b except the central tile are fixed to $|\tilde{1}\rangle$, then it is clear that a bipartition whose minimal surface that cuts through the central tile has a non-trivial area operator. The same exercise can be repeated for other tiles. For this specific code, each bulk state that is fixed to $|\tilde{0}\rangle$ reduces the corresponding tensor to a product state, hence creating a puncture in the network in the order of AdS radius. For more specific properties see Section 6 of \cite{Cao:2020ksw}.

\section{Discussions}
Although stabilizer codes have wide applications in quantum error correction and have reproduced key features of holography, their limitations as toy models are also becoming increasingly apparent. It is known that stabilizer codes cannot reproduce certain features of holography, including the entanglement spectrum (non-flatness), multipartite entanglement\cite{tripartite,markovgap}, correlation functions \cite{Pastawski:2015qua,Jahn:2019nmz,Jahn2022tensornetworkmodels}, and the general expectation of CFTs to have non-local magic\footnote{Note that one can deviate from a stabilizer state on the boundary by encoding a non-stabilizer bulk state. However, the amount of non-stabilizerness injected this way is of $O(1)$ and is insufficient to capture certain holographic features, e.g. the Markov gap \cite{markovgap}, which require non-stabilizerness of $O(1/G_N)$. }. We showed that codes without non-local magic also cannot produce non-trivial area operators, and by extension, capture the features of back-reaction or the quantum extremal surface formula \cite{Engelhardt:2014gca}. 
It further motivates the need to understand holographic magic\cite{White:2020zoz,Goto}, its connection with entanglement, and the construction of a more systematic framework for designing non-stabilizer codes.

Given these developments, let us discuss three potential directions to better understand the emergence of gravity in quantum codes, such as the emergence of Einstein's equations (or the Hamiltonian constraint in the absence of dynamics). 
Here we roughly categorize them as 1) the extensions of the stabilizer formalism, 2) gauge constraints, and 3) approximate quantum error correction. 

Although stabilizer codes have many limitations, the general stabilizer formalism remains a powerful machinery \cite{XS,XP,CWS}.  For example, non-regular XP codes seems to support a non-trivial $Z(M_{\mathcal{C}})$. More precisely, it has a superselection sector associated with the core index. This may be connected with approaches where bulk gauge constraints are imposed \cite{Dolev:2021ofc}. Extensions in this front, including considering stabilizer groups with generators that are not given by local tensor products of operators, can be useful for both QECC and holography. Other formalisms, such as operator algebra QEC \cite{OAQEC1,OAQEC2,STABOAQEC}, are also crucial in building codes that circumvent the no-go theorem. Understanding non-trivial centers and area operators in codes with gauge constraints offers an important but under-explored alternative. This is the case for ordinary bottom-up QECC models of holography~\cite{Dolev:2021ofc}, which may be advantageous for near-term quantum simulations, but is also true of RTN approaches~\cite{Qi:2022lbd,Dong:2023kyr}, which are powerful analytical tools. 
Finally, it is well-accepted that AdS/CFT is best described by approximate QECCs. Therefore, definitions related to the area operator and the RT formula should be modified accordingly.  Ref.\cite{AP2021} made progress in this direction, though instances with non-trivial centers remain to be understood.

\section*{Acknowledgement}
The author thanks Alexander Jahn, Patrick Rall, Jason Pollack, and Yixu Wang for helpful comments and discussions. The author also thanks Ning Bao for (inadvertently) convincing him that it is better to let the note sit on arXiv than in his drawer collecting dust. C.C. acknowledges the support by the Air Force Office of Scientific Research (FA9550-19-1-0360), the National Science Foundation (PHY-1733907), and the Commonwealth Cyber Initiative. The Institute for Quantum Information and Matter is an NSF Physics Frontiers Center.

\appendix
\section{Proof of lemmas}\label{app:1}
\begin{lemma}\label{lemma:1}
If $E\in \mathcal{E}$ is an element of a nice error basis, then the action of any $F\in\mathcal{E}$ maps an eigenstate $|\alpha\rangle$ of $E$ to an eigenstate $|\alpha'\rangle$ of $E$ up to a phase. 
\end{lemma}

\begin{proof}
Suppose $|\alpha\rangle$ be an eigenstate of $E$, such that $E|\alpha\rangle = \gamma |\alpha\rangle$. Since $E^{\dagger}E=I$, it is clear that $E^{\dagger}|\alpha\rangle = \gamma^{-1}|\alpha\rangle$. Now for any operator $F$, $FE=\omega(E,F) EF$, then 
$$F|\alpha\rangle=FE^{\dagger}E|\alpha\rangle=\gamma FE^{\dagger}|\alpha\rangle = \gamma \omega(F,E^{\dagger}) E^{\dagger}F|\alpha\rangle.$$ 
Hence $|\alpha'\rangle=F|\alpha\rangle$ is also an eigenstate of $E$ with eigenvalue $(\gamma\omega(F,E^{\dagger}))^{-1}$. 

Furthermore, it is clear that by inserting $E^{\dagger}E$, $$\langle\alpha|\alpha'\rangle=\langle\alpha|F|\alpha\rangle=\langle\alpha|E^{\dagger}EF|\alpha\rangle = |\gamma|^2 \omega(E,F)\langle\alpha|F|\alpha\rangle$$ which implies that the two states are orthogonal if $|\gamma|^2\omega(E,F)\ne 1$. 
\end{proof}

\begin{lemma}\label{lemma:2}
 Suppose $\mathcal{L}_A$ is a non-trivial area operator in a code that is complementary recoverable in the partition $A,\bar{A}$ and $\mathcal{E}=\{Q_i\}\subset M_C$ is a nice error basis, then there exists a physical representation $\tilde{Q}_i$ such that $[\tilde{Q}_i,\mathcal{L}_A]=0$.
\end{lemma}
\begin{proof}
    Let us first decompose the non-trivial area operator as
 \begin{equation}
     \mathcal{L}_A = c_0 \tilde{I}+\sum_i c_i \tilde{Q}_i,
 \end{equation}
 where there is at least one $\tilde{Q}_i$ for which $c_i\ne 0$. Suppose $\tilde{Q}_i\not\in Z(M_A)$, then $\tilde{Q}_i$ is in $M_A'$ but not contained in $M_A$ or there must be some $\tilde{P}\in M_A$ that doesn't commute with it. For the former, there must be $\tilde{P}'\in M_{\bar{A}}$ that does not commute with it, for which the following argument can be repeated for $A\rightarrow \bar{A}$. For the latter, we have
 $$0=[\tilde{P},\mathcal{L}_A] = \sum_i c_i [\tilde P,\tilde Q_i]= \sum_j d_j [\tilde P,\tilde Q_j]$$
where we have eliminated any $\tilde Q_i$ that commutes with $\tilde P$ from the first sum. Since $PQ_j=\omega(P,Q_j) Q_jP$ where $\omega\ne 1$ by assumption, 
$$\sum_j d_j [\tilde P,\tilde Q_j] = \sum_j d_j (\omega(P,Q_j)-1)\tilde Q_j \tilde P = (\sum_j\alpha_j \tilde Q_j) \tilde P=0.$$
Because the nice error basis forms a group under multiplication, let's write ${R}_j= Q_j  P\in\mathcal{E}$, then it means that $ M=\sum_j\alpha_j  R_j=0$ for $\alpha_j \ne 0$. As $ M$ is null, $\Tr[R^{\dagger}_j M]=0$ for any $R_j\in \mathcal{E}$. However, this contradicts $\Tr[ R_j^{\dagger} M]=\alpha_j\Tr[I]\ne 0$. Therefore there must exist at least one $\tilde{Q}_i\in Z(M_A)$ meaning that $[\tilde{Q}_i,\mathcal{L}_A]=0$. 
\end{proof}

\bibliographystyle{jhep}
\bibliography{ref} 
\end{document}